\documentclass[10pt,a4paper]{article}
\usepackage{comment}
\usepackage{amsmath}
\usepackage{amssymb}
\usepackage{amsthm}
\usepackage{amscd}
\usepackage{graphicx}
\usepackage{indentfirst}
\usepackage[ruled,vlined,linesnumbered]{algorithm2e}
\usepackage{titlesec}
\usepackage[top=25.4mm, bottom=25.4mm, left=31.7mm, right=31.2mm]{geometry}
\usepackage{titlesec}
\usepackage{color}
   \usepackage{amsmath,amssymb}
   \usepackage{latexsym}
\newcommand{\mycite}[1]{$^{\mbox{\rm\scriptsize\cite{#1}}}\!$}
\usepackage[titletoc]{appendix}

\begin{document}
\newtheorem{definition}{\bf Definition}
\newtheorem{theorem}{\bf Theorem}
\newtheorem{proposition}{\bf Proposition}
\newtheorem{lemma}{\bf Lemma}
\newtheorem{corollary}{\bf Corollary}
\newtheorem{example}{\bf Example}
\newtheorem{remark}{\bf Remark}
\newtheorem{Table}{\bf Table}
\newtheorem{Sentence}{\bf Step}
\newtheorem{Branch}{}

\def\T {\ensuremath{{\bf T}}}
\def\N {\ensuremath{{\mathbb N}}}
\def\A {\ensuremath{{\bf A}}}
\def\B {\ensuremath{{\bf B}}}
\def\P {\ensuremath{{\bf P}}}
\def\S {\ensuremath{{\mathbb S}}}
\def\E {\ensuremath{{\bf E}}}
\def\H {\ensuremath{{\bf H}}}
\def\V {\ensuremath{{\rm V}}}
\def\D {\ensuremath{{\rm D}}}
\def\PF {\ensuremath{{\bf PF}}}
\def\TH {\ensuremath{{\bf TH}}}
\def\RS {\ensuremath{{\mathbb T}}}
\def\HCTD {\ensuremath{{\tt HCTD}}}
\def\HPCTD {\ensuremath{{\mathrm HPCTD}}}
\def\CTD {\ensuremath{{\mathrm CTD}}}
\def\PCTD {\ensuremath{{\mathrm PCTD}}}
\def\WUCTD {\ensuremath{{\mathrm WUCTD}}}
\def\RSD {\ensuremath{{\mathrm RSD}}}
\def\FWCTD {\ensuremath{{\mathrm FWCTD}}}
\def\SWCTD {\ensuremath{{\mathrm SWCTD}}}
\def\ARSD {\ensuremath{{\tt RSD}}}
\def\APRSD {\ensuremath{{\tt WRSD}}}
\def\CCTD {\ensuremath{{\tt CTD}}}
\def\ASWCTD {\ensuremath{{\tt SWCTD}}}
\def\ASMPD {\ensuremath{{\tt SMPD}}}
\def\AHPCTD {\ensuremath{{\tt HPCTD}}}
\def\TDU {\ensuremath{{\mathrm TDU}}}
\def\ATDU {\ensuremath{{\tt TDU}}}
\def\RDU {\ensuremath{{\mathrm RDUForZD}}}
\def\ARDU {\ensuremath{{\tt RDUForZD}}}
\newtheorem{Rules}{{\bf Rule}}

\newcommand{\disc}[1]{\mbox{{\rm disc}$(#1)$}}
\newcommand{\alg}[1]{\mbox{{\rm alg}$(#1)$}}
\newcommand{\SAT}[1]{\mbox{{\rm sat}$(#1)$}}
\newcommand{\SQR}[1]{\mbox{{\rm sqrt}$(#1)$}}
\newcommand{\ideal}[1]{\langle#1\rangle}
\newcommand{\I}[1]{\mbox{{\rm I}$_{#1}$}}
\newcommand{\ldeg}[1]{\mbox{{\rm ldeg}$(#1)$}}
\newcommand{\iter}[1]{\mbox{{\rm iter}$(#1)$}}
\newcommand{\mdeg}[1]{\mbox{{\rm mdeg}$(#1)$}}
\newcommand{\lv}[1]{\mbox{{\rm lv}$_{#1}$}}
\newcommand{\mvar}[1]{\mbox{{\rm mvar}$(#1)$}}
\newcommand{\prem}[1]{\mbox{{\rm prem}$(#1)$}}
\newcommand{\pquo}[1]{\mbox{{\rm pquo}$(#1)$}}
\newcommand{\rank}[1]{\mbox{{\rm rank}$(#1)$}}
\newcommand{\res}[1]{\mbox{{\rm res}$(#1)$}}
\newcommand{\cls}[1]{\mbox{{\rm cls}$_{#1}$}}
\newcommand{\sat}[1]{\mbox{{\rm sat}$(#1)$}}
\newcommand{\sep}[1]{\mbox{{\rm sep}$(#1)$}}
\newcommand{\tail}[1]{\mbox{{\rm tail}$(#1)$}}
\newcommand{\zm}[1]{\mbox{{\rm MZero}$(#1)$}}
\newcommand{\zero}[1]{\mbox{{\rm Zero}$(#1)$}}
\newcommand{\rd}[1]{\mbox{{\rm Red}$(#1)$}}
\newcommand{\map}[1]{\mbox{{\rm map}$(#1)$}}
\newcommand{\op}[1]{\mbox{{\rm op}$(#1)$}}

\title{ {Generic Regular Decompositions for Parametric Polynomial Systems}}
\author{Zhenghong Chen\thanks{Corresponding author Email:chenzhenghong@pku.edu.cn}\hspace{1.0em}Xiaoxian Tang\hspace{1.0em}Bican Xia\\
         {\small LMAM \& School of Mathematical Sciences}\\
         {\small Peking University,  Beijing 100871,  China}\\
         {\small chenzhenghong@pku.edu.cn,\ \ tangxiaoxian@pku.edu.cn, \ \ xbc@math.pku.edu.cn}}
\date{}
\maketitle
\begin{abstract}

This paper presents a generalization of our earlier work in \cite{tang}. In this paper, the two concepts,  generic regular decomposition (GRD) and regular-decomposition-unstable (RDU) variety introduced in \cite{tang} for generic zero-dimensional systems, are extended to the case where the parametric systems are not necessarily zero-dimensional. 
An algorithm is provided to compute GRDs and the associated RDU varieties of parametric systems simultaneously on the basis of the algorithm for generic zero-dimensional systems proposed in \cite{tang}.  Then the solutions of any parametric system can be represented by the solutions of finitely many regular systems and the decomposition is stable at any parameter value in the complement of the associated RDU variety of the parameter space. The related definitions and the results presented in \cite{tang} are also generalized and a further discussion on RDU varieties is given from an experimental point of view.  The new algorithm has been implemented on the basis of DISCOVERER \cite{discover} with Maple 16 and experimented with a number of benchmarks from the literature. 

~
\end{abstract}

\section{Introduction}\label{SecIntro}

 As is well known, solving parametric polynomial system plays a key role in many application fields such as automated geometry theorem deduction,  stability analysis of dynamical systems,  robotics and so on. 
 To solve a parametric system symbolically, a basic idea is to transform the system into new systems with special structures or properties so that the solutions of the original system can be handled via studying the solutions of the new systems, which is relatively easy.
Remarkable examples of such methods are the algorithms for computing {\em comprehensive Gr\"obner systems} (CGS) and {\em comprehensive Gr\"obner bases} (CGB) \cite{CGS, sun,Montes,KN,SS,newSS}. The methods based on {\em triangular decompositions} are another kind of such examples \cite{marco, changbo, gxs1992, kalk, maza, dkwang, wangi, wu, yhx01, xia, zjzi}.


Since Wu's work \cite{wu}, lots of well-known methods based on triangular decompositions have been proposed. An essential concept, ``regular chain" (or ``normal chain"), and  algorithms for computing regular chain decomposition have been introduced by Kalkbrener \cite{kalk} and Yang and Zhang \cite{zjzi} independently. For parametric systems, Gao and Chou proposed a method in \cite{gxs1992}  for identifying all parametric values for which a given system has solutions and giving the solutions by $p-$chains 
without a partition of the parameter space.
Wang gave an efficient algorithm for computing regular system decomposition \cite{wangi,  wang, wangEpsilon}, which is a generalization of regular chain decomposition. 
The concept of comprehensive triangular decomposition (CTD) introduced by Chen {\it et al.} in \cite{changbo} is an analogue of the CGS for solving parametric polynomial systems. 

Two new concepts,   generic regular decomposition and regular-decomposition-unstable (RDU) variety for generic zero-dimensional  systems,  are introduced in \cite{tang} and an algorithm is proposed for computing a generic regular decomposition and the associated RDU variety of a given generic zero-dimensional system simultaneously. The solutions of the given system can be expressed by finitely many zero-dimensional regular chains if the parameter value is not on the RDU variety.  The so-called weakly relatively simplicial decomposition (WRSD) plays a crucial role in the algorithm,  which is based on the theories of subresultants.

In this paper, the concepts,  generic regular decomposition (GRD) and regular-decomposition-unstable (RDU) variety, introduced in \cite{tang} for generic zero-dimensional systems are extended to the case where the parametric systems are not necessarily zero-dimensional. 
An algorithm is provided to compute GRDs and the associated RDU varieties of parametric systems simultaneously on the base of the algorithm for generic zero-dimensional systems proposed in \cite{tang}.  Then the solutions of any parametric system can be represented by the solutions of finitely many regular systems and the decomposition is stable at any parameter value in the complement of the associated RDU variety of the parameter space. The new algorithm has been implemented on the base of DISCOVERER \cite{discover} with Maple 16 and experimented with a number of benchmarks from the literature \cite{changbo,  zxq, sun,  Montes,  KN}.  Empirical results are also presented to show the good performance of the algorithm.  In other words, this paper presents a generalization of our earlier work in \cite{tang}. 

First of all, we need to introduce the idea proposed in \cite{tang} briefly.  For a given generic zero-dimensional system $\P$ with $n$ variables and $d$ parameters,  we  considered the parameters as ``constants" and proposed Algorithm {\tt RDUForZD} for computing a so-called {\em generic regular decomposition} ${\mathbb T}$
of  $\P$ in $\overline{K(U)}^{n}$  such that ${\rm V}_{\overline{K(U)}}({\bf P})=\cup _{{\T}\in {\mathbb T}}{\rm V}_{\overline{K(U)}}({\bf T}\backslash \I{{\T}})$, where ${\mathbb T}$ is a set of regular chains.
At the same time, the algorithm would obtain a parametric polynomial such that the regular decomposition was {\em stable} at any parametric point outside the variety (called RDU variety) generated by the parametric polynomial. Roughly speaking, ``stable at a parametric point" means that the regular decomposition remains after we substitute the point for the parameters in $\P$ and ${\mathbb T}$ (see Definition \ref{DEsus}). As a result,  the original generic zero-dimensional system is ``solved" except for the case where parameters are on the RDU variety. That is why the decomposition is called {\em generic} regular decomposition.

Now we would like to show some new ideas of this paper. If the given system is not generic zero-dimensional, to obtain a decomposition with similar properties as in the zero-dimensional case, we choose to express the solutions of the system by finitely many {\em regular systems} \cite{wang} instead of regular chains. So we need to generalize the concept ``generic regular (chain) decomposition" introduced in \cite{tang} into ``generic regular (system) decomposition" (see Definitions \ref{regsys},
\ref{DeRSSwell} and \ref{DEsus}). For solving a positive dimensional system, a natural idea is to view some variables as parameters and call recursively the algorithm for generic zero-dimensional systems proposed in \cite{tang}. However, to prove the correctness of this procedure, we need to study the properties of characteristic sets under specifications carefully (see Lemma \ref{wuset} and Corollary \ref{wudecom}). Besides, it is worth to notice that we have two different interpretations for the results computed by Algorithm {\tt ZDToRC} and both of them play a key role in the proof of the correctness (see Lemma \ref{LEpositive}).  Finally, we give an algorithm which, for any parametric system $\P$, computes a finite set $\mathbb{TH}$  of regular systems in $K[U][X]$ and a polynomial $B\in K[U]$, such that
\begin{enumerate}
  \item ${\V}_{\overline{K(U)}}({\P})=\cup_{[{\T}, H]\in {\mathbb{TH}}}{\V}_{\overline{K(U)}}({\T}\backslash H)$; and
  \item for any $a\in \overline{K}^d\backslash {\V}^{U}(B)$, ${\V}({\P}(a))=\cup_{[{\T}, H]\in {\mathbb{TH}}}{\V}({\T}(a)\backslash H(a))$ and $[{\T},H]$ specializes well at $a$ for any $[{\T},H]\in {\mathbb{TH}}$.
\end{enumerate}
Please see Algorithm \ref{ALsus} in this paper for more details.
What's more, for different orderings of variables, the efficiency of the algorithm can be different and the RDU varieties can be totally different.

At the end of this section,  it is worth to point out that the algorithm provided in this paper has a different feature compared to some existing algorithms. The algorithm for computing regular system decomposition proposed in  \cite{wangi,wang} uses the so-called {\em variable elimination}, which computes a main branch at first and then gets the other branches one by one.  An {\em incremental algorithm}, introduced in \cite{changbo,changbo2011} for computing regular chain decomposition, computes a regular chain decomposition for some polynomials in the given system at first and then intersects the other polynomials with the regular chains one by one.  The algorithm proposed in this paper makes use of a {\em hierarchical strategy}. From an experimental point of view, different strategies are suitable for different benchmarks (see Section \ref{sectionexamples}).


The paper is organized as follows. Section \ref{sfuhaoshuoming} provides basic definitions and concepts that are needed to understand the main algorithm. Section \ref{sectionus} contains the main algorithm,  namely Algorithm \ref{ALsus},  and some relative subalgorithms. Also we review the description of algorithms in our former article. Besides,  proofs for these algorithms are presented in this section.  Some illustrative examples, the empirical data and comparison with previous work along with several implementation details are presented in Section \ref{sectionexamples}.  Section \ref{con1} concludes the paper with a discussion on our future work along this direction.

\section{Preliminaries}\label{sfuhaoshuoming}

The following paragraphs give a brief outline of the vocabulary and tools we will be using throughout the paper. All concepts without precise definitions can be found in \cite{cox, wu, xia}. $\mathbb R$ and $\mathbb C$ stand for the feild of real numbers and the field of complex numbers,  respectively.

         Suppose $\{u_1,  \ldots,  u_d,  x_1,  \ldots,  x_n\}$ is a set of indeterminates with a given order $u_1\prec. . . \prec u_d\prec x_1\prec. . . \prec x_n$
        where $\{u_1,  \ldots,  u_d\}$ and $\{x_1,  \ldots,  x_n\}$ are the sets of parameters and variables,  respectively.
       Let $U=\{u_1,  \ldots,  u_d\}$ and $X=\{x_1,  \ldots,  x_n\}$.
       Suppose $K$ is a field and
        $\overline {K}$ denotes its algebraic closure. Let $K[U]$ be the ring of polynomials in $U$ with coefficients in $K$ and $K(U)$ be the rational function field. A non-empty finite subset $\P$ of $K[U][X]$ is said to be a {\em polynomial system} or {\em system}. If ${\P}\subset K[U][X]\backslash K[X]$,  it is a {\em parametric polynomial system} or {\em parametric system}. If ${\P}\subset K[X]$,  it is a {\em constant polynomial system} or {\em constant system}.

        For a non-empty finite subset $\P\subset$$K[U][X]$ ($\overline{K}[X]$),  $\ideal{{\P}}_{K[U][X]}$ ($\ideal{{\P}}_{\overline{K}[X]}$) denotes the ideal generated by ${\P}$ in $K[U][X]$ ($\overline{K}[X]$) and $\sqrt{\ideal{{\P}}_{K[U][X]}}$ ($\sqrt{\ideal{{\P}}_{\overline{K}[X]}}$) denotes the radical ideal of $\ideal{{\P}}_{K[U][X]}$ ($\ideal{{\P}}_{\overline{K}[X]}$).
        For any $F$ in $K[U][X]\backslash \{0\}$ ($\overline{K}[X]\backslash \{0\}$) and for any $x\in X$,  if $x$ appears in $F$,   $F$ can be regarded as a univariate polynomial in $x$,  namely $F=C_0x^m+C_1x^{m-1}+\ldots+C_m$ where $C_0, C_1, \ldots, C_m$ are polynomials in $K[U][X\backslash \{x\}]$ ($\overline{K}[X\backslash \{x\}]$) and $C_0\neq 0$. Then $m$ is the {\em leading degree} of $F$ {\it w.r.t.} $x$ and is denoted by $\deg(F, x)$. Note that if $x$ does not appear in $F$,  $\deg(F, x)=0$. If there exists $p$ $(1\leq p\leq n)$ such that $\deg(F, x_p)>0$ and for every $i$ $(p<i\leq n)$,  $\deg(F, x_i)=0$,  then the {\em class} of $F$ is $p$. If $\deg(F, x_i)=0$ for every $i$ ($1\leq i\leq n$),   then the {\em class} of $F$ is $0$. The class of $F$ in $K[U][X]\backslash \{0\}$ ($\overline{K}[X]\backslash \{0\}$) is denoted by $\cls{F}$. If $\cls{F}>0$,  $x_{\cls{F}}$ is the {\em main variable} of $F$ and is denoted by $\mvar{F}$. Assume that $F=C_0x_p^m+C_1x_p^{m-1}+\ldots+C_m$ where $p=\cls{F}>0$ and $C_0\neq 0$,  then $C_0$, denoted by $\I{F}$, is the {\em initial} of $F$ and $x_p^m$, denoted by $\rank{F}$, is the {\em rank} of $F$.

 For any ${\bf P}\subset K[U][X]$,
$\V_{\overline{K}}(\bf P)$ denotes the set
$\{(a_1,  \ldots,  a_{d+n})\in \overline{K}^{d+n}|P(a_1,  \ldots,  a_{d+n})=0,  \forall P\in \P\}$. For any ${\bf B}\subset K[U]$,  ${\rm V}^{U}({\bf B})$ denotes the set
$\{(a_1,  \ldots,  a_d)\in \overline{K}^{d}|B(a_1,  \ldots,  a_d)=0,  \forall B\in \B\}.$
And ${\rm V}_{\overline{K(U)}}({\bf P})$ denotes the set
$\{(a_1,  \ldots,  a_n)\in \overline{K(U)}^{n}|P(U, a_1,  \ldots,  a_n)=0,  \forall P\in \P\}$.
        For any $F\in K[U][X]$,
  the coefficients $B_1,  \ldots,  B_t$ of $F$  in $X$ are polynomials in $K[U]$,  then
  ${\rm V}^{U}(F)$ denotes ${\rm V}^{U}(\{B_1,  \ldots,  B_t\})$.
   Note that for two finite subsets ${\bf P}$ and ${\bf H}$ of $K[U][X]$,
        ${\rm V}_{\overline{K(U)}}({\bf P}\backslash {\bf H})$ denotes the set
        ${\rm V}_{\overline{K(U)}}({\bf P})\backslash {\rm V}_{\overline{K(U)}}({\bf H})$. Similarly,  we can have $\V({\bf P}\backslash {\bf H})$,  $\V_{\overline{K}}({\bf P}\backslash {\bf H})$ and ${\rm V}^{U}({\bf P}\backslash {\bf H})$.
         $\dim(\mathcal{V})$ denotes the {\em dimension} of an {\em affine variety} $\mathcal{V}$.


         \textbf{\textsf{Triangular Set}}.           A non-empty finite set ${\T}=\{T_1, T_2\ldots, T_r\}$ of polynomials in $K[U][X]$ ($\overline{K}[X]$)  is  a {\em triangular set} in $K[U][X]$ ($\overline{K}[X]$) if $0<\cls{T_1}<\cls{T_2}<\ldots<\cls{T_r}$. For a triangular set $\T$ in $K[U][X]$ ($\overline{K}[X]$),  $\I{{\T}}$,  $\mvar{{\T}}$  and $\rank{{\T}}$ denote $\Pi_{T\in {\T}}\I{T}$,  $\{\mvar{T}|T\in {\T}\}$ and $\{\rank{T}|T\in {\T}\}$, respectively. The {\em saturated ideal} of a triangular set $\T$ in $K[U][X]$ is defined as the set $\{F\in K[U][X]|\I{\T}^sF\in \ideal{{\T}}_{K[U][X]}$ for some positive integer $s\}$ and is denoted by $\SAT{\T}_{K[U][X]}$. Similarly,  the {\em saturated ideal} of a triangular set $\T$ in  $\overline{K}[X]$ is defined as the set $\{F\in \overline{K}[X]|\I{\T}^sF\in \ideal{{\T}}_{\overline{K}[X]}$ for some positive integer $s\}$ and is denoted by $\SAT{\T}_{\overline{K}[X]}$.
 Suppose $F\in K[U][X]$ ($\overline{K}[X]$) and $\T$ is a triangular set in $K[U][X]$  ($\overline{K}[X]$),  then $F$ is {\em reduced} {\it w.r.t.} $\T$ if $\deg(F, \mvar{T_i})<\deg(T_i, \mvar{T_i})$ for every $i$ $(1\leq i\leq r)$. A triangular set ${\T}=\{T_1, T_2\ldots, T_r\}$ in $K[U][X]$ ($\overline{K}[X]$) is a {\em non-contradictory ascending chain} in $K[U][X]$ ($\overline{K}[X]$) if $T_i$ is reduced {\it w.r.t.} $\{T_1, \ldots, T_{i-1}\}$ for every $i$ $(2\leq i\leq r)$. A single-element set $\{F\}\subset K[U]$ ($\{F\}\subset \overline{K}$) is  a {\em contradictory ascending chain} in $K[U][X]$ ($\overline{K}[X]$) if $F\neq 0$. Remark that an {\em ascending chain} is either a non-contradictory ascending chain or a contradictory ascending chain.


        \textbf{\textsf{Successive Pseudo Remainder}}.            For two polynomials $F$ and $P$ in $K[U][X]$ ($\overline{K}[X]$) and a variable $x\in X$,  the {\em pseudo remainder} of $F$ {\em pseudo-divided} by $P$ {\it w.r.t.} $x$ is
        denoted by $\prem{F,  P,  x}$. Particularly, $\prem{F, P, \mvar{P}}$ is denoted by $\prem{F, P}$.
        For a polynomial $F\in K[U][X]$ ($\overline{K}[X]$) and a triangular set ${\T}=\{T_1,  . . . ,  T_r\}$ in $K[U][X]$ ($\overline{K}[X]$),
        the {\em successive pseudo remainder} \cite{zjzi} of $F$ {\it w.r.t.} ${\T}$
        is denoted by $\prem{F,  {\T}}$,  namely
\[\prem{F,  {\T}}=\prem{\ldots\prem{\prem{F,T_r},T_{r-1}},\ldots,T_1}.\]
        For a finite set ${\P}\subset K[U][X]$ ($\overline{K}[X]$),  $\prem{{\P},  {\T}}$
        denotes the set $\{\prem{F,  {\T}}\mid F\in {\P}\}$.

        \textbf{\textsf{Successive Resultant}}.  For two polynomials $F$ and $P$ in $K[U][X]$ ($\overline{K}[X]$) and a variable $x\in X$,  the {\em resultant} \cite{zjzi} of $F$ and $P$ {\it w.r.t.} $x$ is
        denoted by $\res{F,  P,  x}$. Particularly,  $\res{F, P, \mvar{P}}$ is denoted by $\res{F, P}$.
        For a polynomial $F\in K[U][X]$ ($\overline{K}[X]$) and a triangular set ${\T}=\{T_1,  . . . ,  T_r\}$ in $K[U][X]$ ($\overline{K}[X]$),
        the {\em successive resultant} \cite{zjzi} of $F$ {\it w.r.t.} ${\T}$
        is denoted by $\res{F,  {\T}}$,  namely
\[\res{F,  {\T}}=\res{\ldots\res{\res{F,T_r},T_{r-1}},\ldots,T_1}.\]

       \textbf{\textsf{Regular Chain}}.   A triangular set ${\bf T}=\{T_1,  \ldots,  T_r\}$ in
       $K[U][X]$ ($\overline{K}[X]$) is said to be a {\em regular chain} in $K[U][X]$ ($\overline{K}[X]$),   if $\I{{T_1}}\neq 0$ and for each $i$ $(1<i\leq r)$,  $\res{\I{T_i},  \{T_{i-1},  \ldots,  T_1\}}\neq 0$.
       If ${\bf T}$ is a regular chain in $K[U][X]$ ($\overline{K}[X]$) and $\mvar{{\T}}=X$,  $\T$ is a {\em zero-dimensional regular chain}.

       \textbf{\textsf{Regular System}}\footnote{The definition of regular system is the same as that introduced in \cite{changbo} and different from that proposed in \cite{wangi}, see more details in \cite{changbo}. }.\mycite{changbo}  Let ${\T} \subset K[U][X]$($\overline{K}[X]$) be a regular chain and $H \in K[U][X](\overline{K}[X])$. If  ${\rm res}(H,\T) \neq 0$, then $[{\T} ,H]$ is said to be a regular system in $K[U][X]$($\overline{K}[X]$).

\begin{proposition}\mycite{changbo}
       If $[{\T},H]$ is a regular system in $K[U][X]$, then $\V_{\overline{K(U)}}({\T} \backslash H) \neq \emptyset$.
\end{proposition}

       \textbf{\textsf{Assignment Homomorphism}}.        For each $a=(a_1,  \ldots,  a_d)\in {\overline{K}}^{d}$,   $\phi
        _{a}:K[U][X]\longrightarrow \overline{K}[X]$ is a homomorphism such  that $\phi_{a}(F)=F(a,  X)$ for all
        $F\in K[U][X]$ and we denote $\phi_{a}(F)$ by $F(a)$. For a non-empty finite set ${\P}\subset K[U][X]$,  ${\P}(a)$ denotes the set $\{F(a)|F\in{\P}\}$ and remark that ${\P}(a)=\emptyset$ if ${\P}=\emptyset$.

         \textbf{\textsf{Characteristic Set And Wu's Method}}.
       An ascending chain ${\bf C}$ in $K[U][X]$  is a {\em characteristic set} of $\P$ in $K[U][X]$ if ${\bf C}\subset \ideal{{\P}}_{K[U][X]}$ and $\prem{{\P}, {\bf C}}=\{0\}$.  Theorem \ref{wellorder} below is the so called {\em well-ordering principle}.

\begin{theorem}\label{wellorder}\mycite{wu}
There exists an algorithm which, for an input non-empty finite subset ${\P}\subset K[U][X]$, outputs either a contradictory ascending chain meaning that ${\V}_{\overline{K(U)}}({\P})=\emptyset$, or a (non-contradictory) characteristic set ${\bf C}=\{C_1,\ldots,C_t\}$ such that
\[{\V}_{\overline{K(U)}}({\P})={\V}_{\overline{K(U)}}({\bf C}\backslash\I{{\bf C}})\cup\cup_{i=1}^t{\V}_{\overline{K(U)}}({\P}\cup {\bf C}\cup \{\I{C_i}\}).\]
\end{theorem}

On the base of Theorem \ref{wellorder}, there exists an algorithm, namely Wu's method,  for computing a  finite sequence of ascending chains ${\bf C}_1, {\bf C}_2, \ldots, {\bf C}_m$ $(m\geq 1)$ in $K[U][X]$ such that

(1)${\bf C}_1, {\bf C}_2, \ldots, {\bf C}_m$ is a finite sequence of  characteristic sets  in $K[U][X]$;

(2)If $m=1$,  ${\V}_{\overline{K(U)}}({\P})=\emptyset$. Otherwise,  suppose ${\mathbb S}=\{{\bf C}_i|1\leq i\leq m$ and ${\bf C}_i$ is a non-contradictory ascending chain$\}$, then  ${\V}_{\overline{K(U)}}({\P})=\cup_{{\bf C}\in {\mathbb S}}{\V}_{\overline{K(U)}}({\bf C}\backslash\I{{\bf C}})$.

 The set of ascending chains $\{{\bf C}_1, {\bf C}_2, \ldots, {\bf C}_m\}$ above is said to be a {\em Wu's decomposition} or {\em characteristic set decomposition} of $\P$ in $K[U][X]$.

For any  triangular set $\T$ in $K[U][X]$, we denote $\#(X)-\#(\T)$ by $d({\T},X)$.
 \begin{definition}\label{DEdimension}
 Let $\P$ be a parametric system in $K[U][X]$ and the set $ \{{\bf C}_1, {\bf C}_2, \ldots, {\bf C}_m\}$ of ascending chains be a {\em Wu's decomposition} of $\P$ in $K[U][X]$.   If  $d({\bf C}_i,X)=0$ for every non-contradictory ascending chain ${\bf C}_i$, $\P$ is said to be a {\em generic zero-dimensional system}. Otherwise, $\P$ is said to be a {\em generic positive-dimensional system}.
\end{definition}

\begin{definition}\label{regsys}
Let $\P$ be a parametric system in $K[U][X]$ and $\mathbb{TH}=\{[{\T}_1,H_1],\ldots,[{\T}_s,H_s]\}$ be a set of regular systems  in $K[U][X]$. $\mathbb{TH}$ is said to be a {\em parametric regular system decomposition} of $\P$ in $K[U][X]$, if ${\V}_{\overline{K(U)}}({\P})=\cup_{i=1}^s{\V}_{\overline{K(U)}}({\T}_i\backslash H_i)$.
\end{definition}

\begin{definition}\mycite{changbo}\label{DeRSSwell}
     Let $[{\bf T}, H]$ be a regular system in $K[U][X]$ and $a\in\overline{K}^{d}$.   If ${\bf
     T}(a)$ is a regular chain in $\overline{K}[X]$,
     $\rank{{\bf T}(a)}$$=\rank{{\bf T}}$ and $\res{H(a),{\T}(a)} \neq 0$,
     then we say that the regular system $[{\bf T}, H]$ {\em specializes well} at $a$.
\end{definition}


\begin{definition}
Let $\P$ be a parametric system in $K[U][X]$ and $\mathbb{TH}=\{[{\T}_1,H_1],\ldots,[{\T}_s,H_s]\}$ be a parametric regular system decomposition of $\P$ in $K[U][X]$.  For any $a \in \overline{K}^d$, if ${\V}({\P}(a))=\cup_{i=1}^s{\V}({\T}_i(a)\backslash H_i(a)))$ and $[{\T}_i,H_i]$($1\leq i \leq s$) specializes well at $a$, then $\mathbb{TH}$ is said to be {\em stable} at $a$.
\end{definition}
Remark that the concept of {\em stable} generic regular (chain) decomposition is first introduced in \cite{tang}.


\begin{definition}\label{DEsus}
      Let $\mathbb{ TH}$ be a  parametric regular system decomposition of a given parametric system $\P$ in $K[U][X]$. If there is an affine variety $\mathcal{V}$ in $\overline {K}^{d}$ with $\dim(\mathcal{V})<d$ such that $\mathbb{ TH}$ is stable at any $a\in \overline{K}^d\backslash \mathcal{V}$, then $\mathbb{TH}$ is said to be a {\em generic regular system decomposition} of $\P$ and $\mathcal{V}$ is said to be a {\em regular-decomposition-unstable variety (RDU)} of $\P$ {\it w.r.t.} $\mathbb{TH}$.
\end{definition}





\begin{definition}\label{Dewrsd}
Let $\T$ be a zero-dimensional regular chain in $K[U][X]$ and $P\in K[U][X]$. Suppose  $\mathbb{H}$ and $\mathbb{G}$ are two finite sets of zero-dimensional regular chains in $K[U][X]$.  If

$(1)$ $\V_{\overline{K(U)}}({\T}\cup \{P\})=\cup_{{\bf H}\in {\mathbb H}}{\V}_{\overline{K(U)}}({\bf H})$,

$(2)$ $\V_{\overline{K(U)}}({\T}\backslash P)=\cup_{{\bf G}\in {\mathbb G}}{\V}_{\overline{K(U)}}({\bf G})$,\\
then $(\mathbb{H}, \mathbb{G})$ is said to be a {\em weakly relatively simplicial decomposition} of $\T$ {\it w.r.t.} $P$ in $K[U][X]$.
\end{definition}

In \cite{tang}, we gave an algorithm for computing weakly relatively simplicial decompositions. Here, we only present its specification but omit the details.

 \begin{algorithm}\label{ALPRSD}
\DontPrintSemicolon
\SetAlgoCaptionSeparator{. }
 \caption{\APRSD}
    \KwIn{A zero-dimensional regular chain ${\bf T}=\{T_1, \ldots, T_n\}$ in $K[U][X]$,   a polynomial $P\in K[U][X]$,  variables $X=\{x_1, \ldots, x_n\}$ }
    \KwOut{[${\mathbb H}$, ${\mathbb G}$, $B$],  where\\

       (1)$(\mathbb{H}, \mathbb{G})$ is a weakly relatively simplicial decomposition of $\T$ {\it w.r.t.} $P$ in $K[U][X]$;

       (2)$B$ is a polynomial in $K[U]$ such that for any $a\in \overline{K}^d\backslash {\V}^{U}(B)$, the weakly relatively simplicial decomposition $(\mathbb{H}, \mathbb{G})$ of $\T$ {\it w.r.t.} $P$ is stable\footnotemark\  at $a$.}
\end{algorithm}
\footnotetext{Please see the definition of {\em stable} in \cite{tang}.}

\section{Theory and Algorithm}\label{sectionus}

We first give some notations. Assume that ${\tt Alg}$ is a name of an algorithm and $p_1, \ldots, p_t$ is a sequence of inputs of this algorithm. If the output of ${\tt Alg}(p_1, \ldots, p_t)$ is a finite sequence $q_1, \ldots, q_s$,  $q_i$ is denoted by ${\tt Alg}(p_1, \ldots, p_t)_i$ for any $i$ $(1\leq i\leq s)$ and also said to be the $i$th output of ${\tt Alg}(p_1, \ldots, p_t)$. Given a finite set $S=\{s_1, \ldots, s_t\}$ and a map $\phi$ on $S$,  ${\tt op}(S)$ denotes the finite sequence $s_1, \ldots, s_t$  and ${\tt map}(s\rightarrow \phi(s),  S)$ denotes the set $\phi(S)=\{ \phi(s) | s \in S\}$.

\subsection{Wu's Decomposition Under Specification}\label{subsectionwu}
The general idea of Wu's method is presented in Section \ref{sfuhaoshuoming}. 
Some results on Wu's decomposition under specification are given in this section.

\begin{definition}\label{line}
Let ${\P}_1$ be a parametric system in $K[U][X]$ and $\mathbb{S}=\{{\bf C}_1,\ldots,{\bf C}_m\}$ be a Wu's decomposition of ${\P}_1$ in $K[U][X]$. Suppose $\mathcal{L}=\{{\bf C}_{l,1},{\bf C}_{l,2},\ldots,{\bf C}_{l,k} \}$ is a subset of $\mathbb{S}$ satisfying that

 (1) ${\bf C}_{l,1}$ is a characteristic set of ${\P}_1$.

 (2) If $k \geq 2$, ${\bf C}_{l,i}$ $(2\leq i \leq k)$ is a characteristic set of ${\P}_{i}={\P}_{i-1} \cup {\bf C}_{l,i-1} \cup \{\I{C_{l,i-1}}\}$ where $C_{l,i-1} \in {\bf C}_{l,i-1}$.

 (3) If $k =1$, ${\bf C}_{l,1}$ is a contradictory ascending chain. Otherwise, ${\bf C}_{l,i}$ $(1\leq i \leq k-1)$ is a non-contradictory ascending chain and ${\bf C}_{l,k}$ is a contradictory ascending chain.

 Then $\mathcal{L}$ is said to be a {\em line} of $\mathbb{S}$ and $\P^{\mathcal L}=\{ \P_1,\ldots, \P_k \}$ the corresponding systems.
\end{definition}

\begin{lemma}\label{wusinit}
 Let ${\P}_1$ be a parametric system in $K[U][X]$ and $\mathbb{S}=\{{\bf C}_1,\ldots,{\bf C}_m\}$ be a Wu's decomposition of ${\P}_1$ in $K[U][X]$. Let $\mathcal{L}=\{{\bf C}_{l,1},\ldots,{\bf C}_{l,k} \}$ be a line of $\mathbb{S}$ with corresponding systems $\{ {\P}_1, \ldots, {\P}_k \}$. Then for any $a \in \overline{K}^d\backslash {\V}^U({\bf C}_{l,k})$,  there exists a polynomial $p_i \in {\P}_i$ such that $p_i(a) \not\equiv 0$ for any $i$ $(1\leq i \leq k)$.
\end{lemma}

\begin{proof}
We prove it by induction on the number k of elements of $\mathcal{L}$. If k=1, it means $\mathcal{L}$ contains only one element. Since ${\bf C}_{l,1}$ is a contradictory ascending chain, we can assume that ${\bf C}_{l,1}=\{C_{l,1}\}$ where $C_{l,1} \in K[U]$ and  we know that $C_{l,1} \in \ideal{{\P}_1}$. Suppose ${\P}_1=\{f_{1,1},\ldots,f_{1,t_1}\}$. Then $C_{l,1}$ can be written as $C_{l,1}= \sum_{j=1}^{t_1} h_j f_{1,j}$ where $h_j\in K[U][X]$ for any $j$ $(1\leq j\leq t_1)$. Thus $C_{l,1}(a)= \sum_{j=1}^{t_1} h_j(a) f_{1,j}(a)$ holds. Since $C_{l,1}(a) \not\equiv 0$, there must exist some $f_{1,e_1}\in {\P}_1$ such that $f_{1,e_1}(a) \not\equiv 0$. Let $p_1=f_{1,e_1}$ and we are done.

Now we assume that the conclusion holds when $k<N$ ($N>1$). Suppose $k=N$. $\mathcal{L}_2=\{{\bf C}_{l,2},\ldots,{\bf C}_{l,k} \} \subsetneq \mathcal{L}$ is a line of  Wu's decomposition of ${\P}_2$. According to the induction hypothesis, for any $a \in \overline{K}^d\backslash {\V}^U({\bf C}_{l,k})$,  there exists a polynomial $p_i \in {\P}_i$ such that $p_i(a) \not\equiv 0$ for any $i$ $(2\leq i \leq k)$.
As is known to us, ${\P}_{2}={\P}_{1} \cup {\bf C}_{l,1} \cup \{\I{C_{l,1}}\}$, if $p_2\in {\P}_{1}$, let $p_{1}=p_2$ and the conclusion holds obviously. Otherwise,
 $p_{2} \in \{\I{C_{l,1}}\} \cup {\bf C}_{l,1}$.  If  $p_2 \in {\bf C}_{l,1}\subset \ideal{{\P}_{1}}$, supposing ${\P}_{1}=\{f_{1,1},\ldots,f_{1,t_{1}}\}$, then $p_{2}$ can be written as $p_{2}= \sum_{j=1}^{t_{1}} h_j f_{1,j}$ where $h_j \in K[U][X]$ for any $j$ $(1\leq j\leq t_{1})$. Thus $p_2(a)= \sum_{j=1}^{t_{1}} h_j(a) f_{1,j}(a)$ holds. Since $p_2(a) \not\equiv 0$, there must exist some $f_{1,e_{1}}\in {\P}_{1}$ such that $f_{1,e_{1}}(a) \not\equiv 0$. Let $p_{1}=f_{1,e_{1}}$ and we are done. If $p_{2}= \{\I{C_{l,1}}\}$,  it implies  $C_{l,1}(a)\not\equiv 0$. Since $C_{l,1}$ can be written as $C_{l,1}= \sum_{j=1}^{t_{1}} g_j f_{1,j}$ where $g_j\in K[U][X]$ for any $j$ $(1\leq j\leq t_{1})$. Thus there must exist some $f_{1,e_{1}}\in {\P}_{1}$  such that $f_{1,e_{1}}(a) \not\equiv 0$. Let $p_{1}=f_{1,e_{1}}$ and  the conclusion holds.

\end{proof}

With Lemma \ref{wusinit}, Corollary \ref{wuonepoly} holds obviously.

\begin{corollary}\label{wuonepoly}
Let ${\P}$ be a parametric system in  $K[U][X]$ and $\mathbb{S}=\{{\bf C}_1,\ldots,{\bf C}_m\}$ be a Wu's decomposition of ${\P}$ in $K[U][X]$. Let $\mathcal{L}=\{{\bf C}_{l,1},{\bf C}_{l,2},\ldots,{\bf C}_{l,k} \}$ be a line of $\mathbb{S}$.  If ${\P}=\{p\}$, containing only one polynomial in K[U][X], $p(a) \not\equiv 0$ for any $a \in \overline{K}^d\backslash {\V}^U({\bf C}_{l,k})$.
\end{corollary}


\begin{lemma}\label{wuset}
Suppose ${\bf C}=\{C_1,\ldots,C_t\}$ is a non-contradictory ascending chain and a characteristic set of parametric system ${\P}$ in $K[U][X]$. For any $a \in \overline{K}^d$, ${\V}({\P}(a)) = {\V}({\bf C}(a) \backslash \I{{\bf C}}(a))   \cup \cup_{i=1}^t{\V}(( {\P}(a)  \cup {\bf C}(a) \cup\{ \I{C_i}(a) \}))$.
\end{lemma}
\begin{proof}
For any $a \in \overline{K}^d$, if $\I{{\bf C}}(a)\equiv0$, the conclusion holds obviously since $\I{{\bf C}}(a)= \prod_{i=1}^t \I{C_i}(a)$. Now we prove the conclusion when $\I{{\bf C}}(a)\not\equiv 0$. According to the definition of characteristic set, we know that ${\bf C}\subset \ideal{{\P}}$ and for any $p \in {\P}$,
${\I{C_1}}^{k_1}\cdots {\I{C_t}}^{k_t}p=q_1 C_{1} +\cdots +q_t C_{t}$ where $q_i\in K[U][X]$ for any $i$ $(1\leq i \leq t)$. Then ${\bf C}(a) \subset \ideal{{\P}(a)}$ and ${\I{C_1}}(a)^{k_1}\cdots {\I{C_t}}(a)^{k_t}p(a)=q_1(a)C_1(a) +\cdots +q_t(a) C_t(a)$. Therefore
${\V}({\bf C}(a) \backslash \I{{\bf C}}(a))  \subset {\V}({\P}(a) ) \subset {\V}({\bf C}(a))$
and ${\V}({\P}(a)) = {\V}({\bf C}(a) \backslash \I{{\bf C}}(a)) \cup  {\V}({\P}(a)\cup\{\I{{\bf C}}(a)\})$.
Since ${\bf C}(a) \subset \ideal{{\P}(a)}$, ${\V}({\P}(a)\cup \{\I{{\bf C}}(a)\})={\V}({\P}(a)\cup {\bf C}(a) \cup \I{{\bf C}}(a))=\cup_{i=1}^t{\V}(( {\P}(a)  \cup {\bf C}(a) \cup\{ \I{C_i}(a) \}))$. We are done.
\end{proof}

According to  Theorem \ref{wellorder} and Lemma \ref{wuset}, we can get Corollary \ref{wudecom} easily.

\begin{corollary}\label{wudecom}
  Let ${\P}$ be a parametric system in  $K[U][X]$ and  $\{{\bf C}_1,\ldots,{\bf C}_m\}$ be  a Wu's decomposition of ${\P}$ in $K[U][X]$. Suppose $\mathbb{S}= \{ {\bf C}_i|1\leq i \leq m$ and ${\bf C}_i$ is a non-contradictory ascending chain$\}$ and $\mathbb{CS}= \{ {\bf C}_i|1\leq i \leq m$ and ${\bf C}_i$ is a contradictory ascending chain$\}$.
  Then for any $a \in \overline{K}^d\backslash(\cup_{{\bf CS} \in \mathbb{CS}} {\V}^U({\bf CS}))$, ${\V}({\P}(a)) =\cup_{{\bf C} \in \mathbb{S}} {\V}({\bf C}(a) \backslash \I{{\bf C}}(a))$.
\end{corollary}



\subsection{Converting To Regular Systems}

  Let ${\T}$ be a triangular set in $K[U][X]$. We can compute a set of finite sequence of regular systems $\mathbb{TH}=\{[{\bf T}_1,H_1],\ldots,[{\bf T}_s,H_s]\}$ in $K[U][X]$ and a polynomial $B\in K[U]$ on the basis of Algorithm \ref{ALPRSD} such that ${\V}_{\overline{K(U)}}({\T}\backslash \I{\T})=\cup_{i=1}^s{\V}_{\overline{K(U)}}({\bf T}_i \backslash H_i)$ and for any $a\in \overline{K}^d\backslash {\V}^{U}(B)$, $\mathbb{TH}$ specializes well at $a$ and ${\V}({\T}(a)\backslash \I{\T}(a))=\cup_{i=1}^s{\V}({\bf T}_i(a) \backslash H_i(a))$. The algorithm is presented as Algorithm \ref{ALpdtors}, which plays a key role in Algorithm \ref{ALsus} proposed in the next section.

  Algorithm \ref{ALzdtorc} below was proposed in \cite{tang} for zero-dimensional case. We just give its specification here.  Our focus in this paper is how to deal with the case where the triangular set is positive-dimensional. So, we need to convert a triangular set with $\mvar {\T} \subsetneq {\bf X}$ to a set of regular systems. 

\begin{algorithm}\label{ALzdtorc}
\SetAlgoCaptionSeparator{. }
\caption{\tt{ZDToRC}}
\DontPrintSemicolon
\KwIn{A triangular set ${\bf T}$ in $K[U][X]$ with $\mvar{{\T}}=X$, variables $X=\{x_1, \ldots, x_n\}$. }
\KwOut{[$\mathbb{G}$, $B$], where\\
             (1)$\mathbb{G}$  is a  finite set  of zero-dimensional regular chains in $K[U][X]$ such that  ${\V}_{\overline{K(U)}}({\bf T}\backslash \I{{\T}})=\cup_{{\bf G}\in \mathbb{G}}{\V}_{\overline{K(U)}}({\bf G})$;\\
            (2)$B$ is a polynomial in $K[U]$ such that for any $a\in \overline{K}^{d}\backslash {\V}^{U}(B)$, ${\V}({\T}(a)\backslash \I{{\T}}(a))=\cup _{{\bf G}\in \mathbb{G}}{\V}({\bf G}(a))$ and ${\bf G}$ specializes well at $a$ for any ${\bf G}\in \mathbb{G}$.}
\end{algorithm}

According to Algorithm \ref{ALzdtorc}, the following Proposition \ref{pro2} is clear.

  \begin{proposition}\label{pro2}
  Let ${\bf T}$ be a triangular set in $K[U][X]$ and ${\tt ZDToRC}({\T},\mvar{\T})=[\mathbb{G}, B]$. If $\mathbb{G} \neq \emptyset$, $\I{\T}(a) \neq 0$ for $a \notin {\V}^U (B)$.
  \end{proposition}

 Now suppose ${\T}=\{T_1,\ldots,T_l\}$ is a triangular set in $K[U][X]$ and $\mvar{\T}\subsetneq X$ where variables $X=\{x_1, \ldots, x_n\}$ and parameters $U=\{u_1, \ldots, u_d\}$.  It is interesting to show that we have two versions to interpret the relationship between ${\T}$ and the results computed by ${\tt ZDToRC}({\T},\mvar{\T})$. Assume that ${\tt ZDToRC}({\T},\mvar{\T})=[\mathbb{G}, B]$. Let $\mathcal{B}({\T})=\mvar{\T}$ and $\mathcal{F}({\T})=X\backslash \mvar{\T}$.
On one hand,  ${\T}$ can be regarded as a triangular set  in $K[U,\mathcal{F}({\T})][\mathcal{B}({\T})]$. At this point, according to Algorithm \ref{ALzdtorc}, we know that \\
(1)${\V}_{\overline{K(U,\mathcal{B}({\T}))}}({\bf T}\backslash \I{{\T}})=\cup
             _{{\bf G}\in \mathbb{G}}{\V}_{\overline{K(U,\mathcal{F}({\T}))}}({\bf G})$,\\
             (2)for any $a\in \overline{K}^{d+n-l}\backslash {\V}^{U,\mathcal{F}({\T})}(B)$, ${\V}({\T}(a)\backslash \I{{\T}}(a))=\cup _{{\bf G}\in \mathbb{G}}{\V}({\bf G}(a))$ and ${\bf G}$ specializes well at $a$ for any ${\bf G}\in \mathbb{G}$ if $\mathbb{G}\neq \emptyset$.

On the other hand, ${\T}$ can also be regarded as a triangular set  in $K[U][\mathcal{F}({\T})][\mathcal{B}({\T})]$. Let $\mathcal{K}=K(U)$, $\mathcal{U}=\mathcal{F}({\T})$ and $\mathcal{X}=\mathcal{B}({\T})$. Then according to Algorithm \ref{ALzdtorc},\\
(3)${\V}_{\overline{\mathcal{K}(\mathcal{U})}}({\bf T}\backslash \I{{\T}})=\cup
             _{{\bf G}\in \mathbb{G}}{\V}_{\overline{\mathcal{K}(\mathcal{U})}}({\bf G})$,\\
(4)for any $a\in \overline{\mathcal{K}}^{n-l}\backslash {\V}^{\mathcal{U}}(B)$, ${\V}_{\overline{\mathcal{K}}}({\T}(a)\backslash \I{{\T}}(a))=\cup _{{\bf G}\in \mathbb{G}}{\V}_{\overline{\mathcal{K}}}({\bf G}(a))$ and ${\bf G}$ specializes well at $a$ for any ${\bf G}\in \mathbb{G}$ if $\mathbb{G}\neq \emptyset$.\\
Remark that  the above statements (1) and (3) are exactly the same since $K(U,\mathcal{U})=\mathcal{K}(\mathcal{U})$. As discussed above, we have the following lemma  by statements (2) and (4)

\begin{lemma}\label{LEpositive}
Suppose ${\T}$ is a triangular set in $K[U][X]$ and ${\tt ZDToRC}({\T},\mvar{\T})=[\mathbb{G}, B]$. Then ${\V}_{\overline{K(U)}}({\T}\backslash \I{{\T}}\cdot B)=\cup _{{\bf G}\in \mathbb{G}}{\V}_{\overline{K(U)}}({\bf G}\backslash B)$ and  ${\V}({\T}(a)\backslash \I{{\T}}(a)\cdot B(a))=\cup _{{\bf G}\in \mathbb{G}}{\V}({\bf G}(a)\backslash B(a))$ for any $a\in \overline{K}^d\backslash {\V}^U(B)$.
\end{lemma}


\begin{algorithm}\label{ALpdtors}
\SetAlgoCaptionSeparator{. }
\caption{\tt{TSToRS}}
\DontPrintSemicolon
\KwIn{A triangular set ${\bf T}=\{T_1, \ldots, T_l\}\subset K[U][X]$, variables $X=\{x_1,\ldots,x_n\}$. }
\KwOut{ [$\mathbb{G}$,$B$], where\\
             (1) $\mathbb{G}$ is a  finite set  of  regular systems in $K[U][X]$ such that     ${\V}_{\overline{K(U)}}({\bf T}\backslash \I{{\T}})=\cup
             _{[{\bf G},H]\in {\mathbb G}}{\V}_{\overline{K(U)}}({\bf G}\backslash H)$;\\
             (2) $B$ is a polynomial  in $ K[U]$, for any $a\in \overline{K}^{d}\backslash {\V}^{U}(B)$, ${\V}({\T}(a)\backslash \I{{\T}}(a))=\cup _{[{\bf G},H]\in {\mathbb G}}{\V}({\bf G}(a)\backslash H(a))$ and
             $[{\bf G},H]$ specializes well at $a$ for any $[{\bf G},H]\in \mathbb{G}$ if $\mathbb{G}\neq \emptyset$. }
           $W$:={\tt ZDToRC}(${\T}, \mvar{\T}$)\;
           $\mathbb{G}$:=$\map{t\rightarrow [t, W_2], W_1}$\;
           \If {$W_2 \in K[U]$}
           {return [$\mathbb{G}$, $W_2$]\;}
           Compute a Wu's decomposition  $\{{\bf C}_1,{\bf C}_2,\ldots,{\bf C}_m\}$  of $\{W_2\}$ in $K[U][X]$ by Wu's method\;
            $B$:=1\;
            \For{$i = 1 \to m$}
         {\eIf{${\bf C}_i$ is a contradictory ascending chain }{ $B$:=$B\cdot {\tt op}({\bf C}_i)$}{
         ${\T}_i:={\T} \cup {\bf C}_i$\;
         $\mathbb{G}:=\mathbb{G} \cup {\tt TSToRS}({\T}_i,X)_1$, $B:=B\cdot {\tt TSToRS}({\T}_i,X)_2$}}

         return $[\mathbb{G},B]$
\end{algorithm}

\begin{theorem}\label{THforALpdtors}
Algorithm \ref{ALpdtors} terminates correctly.
\end{theorem}

\begin{proof}
For a given triangular set ${\bf T}=\{T_1, \ldots, T_l\}$ in $K[U][X]$. Suppose  ${\tt ZDToRC}({\T},\mvar{\T})=[\mathbb{G}_0, B_0]$.

Firstly, we prove Algorithm \ref{ALpdtors} terminates.  If $B_0 \in K[U]$, it terminates obviously. Otherwise, $B_0 \in K[U][\mathcal{F}({\T})]\subset K[U][X]$. Assume that $\{{\bf C}_1,{\bf C}_2,\ldots,{\bf C}_m\}$ is the Wu's decomposition of $\{B_0\}$ in $K[U][X]$ computed by Wu's method and $\mathbb{S}= \{ {\bf C}_i|1\leq i \leq m$ and ${\bf C}_i$ is a non-contradictory ascending chain$\}$.  Obviously $m>1$ since $B_0 \in K[U][\mathcal{F}({\T})]\backslash K[U]$.  For any ${\bf C}_i \in \mathbb{S}$, let ${\T}_i={\T} \cup {\bf C}_i$. Then we know that $\mvar{\T} \subsetneq \mvar{{\T}_i}$, which means $d({\T}_i,X)<d({\T}, X)$. Therefore, it is clearly that Algorithm \ref{ALpdtors} terminates within finite steps of recursion.

Now we prove the correctness by induction on the recursive depth $h$.  If $h=1$, according to Algorithm \ref{ALpdtors}, $B_0 \in K[U]$ and the conclusion follows from Lemma \ref{LEpositive}. Assume that the conclusion holds for $h <N$ ($N > 1$). When $h=N$, $B_0 \in K[U][\mathcal{F}({\T})]\backslash K[U]$.  We can assume that $\{{\bf C}_1,{\bf C}_2,\ldots,{\bf C}_m\}$ is the Wu's decomposition of $\{B_0\}$ in $K[U][X]$ computed by Line 5 in Algorithm \ref{ALpdtors}. Let $\mathcal{N}= \{ i|1\leq i \leq m$ and ${\bf C}_i$ is a non-contradictory ascending chain$\}$ and $\mathcal{CN}=\{1,2,\ldots,m\} \backslash \mathcal{N}$. For any $i \in \mathcal{N}$, let ${\T}_i={\T} \cup {\bf C}_i$ and then ${\T}_i$ is a triangular set in $K[U][X]$. Suppose ${\bf C}_j=\{C_j\}$ ($j \in \mathcal{CN}$), ${\tt TSToRS}({\T}_1,\mvar{{\T}_1})=[\mathbb{G}_i, B_i]$ for any $i$  ($i \in \mathcal{N}$).  Then according to Algorithm \ref{ALpdtors}, ${\mathbb G}=\cup_{{\bf G}
\in {\mathbb G}_0}\{[{\bf G}, B_0]\}\cup \cup_{i \in \mathcal{N}}{\mathbb G}_i$ and $B=\prod_{j \in \mathcal{CN}}C_j\cdot \prod_{i\in \mathcal{N}}B_i.$ By Lemma \ref{LEpositive}, Wu's method and the induction hypothesis, we get
  \begin{equation*}
\begin{split}
&{\V}_{\overline{K(U)}}({\T} \backslash \I{\T})\\
&={\V}_{\overline{K(U)}} ({\T} \backslash \I{\T}\cdot B_0) \cup {\V}_{\overline{K(U)}}({\T}\cup \{B_0\}\backslash \I{\T})\\
&=\cup_{{\bf G}\in {\mathbb G}_0}{\V}_{\overline{K(U)}}({\bf G}\backslash B_0) \cup( {\V}_{\overline{K(U)}}({\T}\backslash \I{\T})\cap {\V}_{\overline{K(U)}}(B_0))\\
&=\cup_{{\bf G}\in {\mathbb G}_0}{\V}_{\overline{K(U)}}({\bf G}\backslash B_0) \cup( {\V}_{\overline{K(U)}}({\T}\backslash \I{\T})\cap \cup_{i\in \mathcal{N}}{\V}_{\overline{K(U)}}({\bf C}_i\backslash \I{{\bf C}_i}))\\
&=\cup_{{\bf G}\in {\mathbb G}_0}{\V}_{\overline{K(U)}}({\bf G}\backslash B_0) \cup( \cup_{i\in \mathcal{N}}{\V}_{\overline{K(U)}}({{\T}_i}\backslash \I{{\T}_i}))\\
&=\cup_{{\bf G}\in {\mathbb G}_0}{\V}_{\overline{K(U)}}({\bf G}\backslash B_0) \cup( \cup_{i\in \mathcal{N}}\cup_{[{\bf G}, H]\in {\mathbb G}_i}{\V}_{\overline{K(U)}}({\bf G}\backslash H))).\\
\end{split}
\end{equation*}
Therefore,  the statement (1) in the specification of Algorithm \ref{ALpdtors} holds.

 For any $a\in \overline{K}^d\backslash {\V}^U (B)$, $C_j(a)\neq 0$ and $B_i(a)\neq 0$ for any $i\in \mathcal{N}$ and $j\in \mathcal{CN}$. Thus by Corollary \ref{wudecom},  ${\V}(B_0(a))=\cup_{i\in \mathcal{N}}{\V}({\bf C}_i(a)\backslash \I{{\bf C}_i}(a))$. By Lemma \ref{LEpositive} and  the induction hypothesis, we get
\begin{equation*}
\begin{split}
&{\V}({\T}(a) \backslash \I{\T}(a))\\
&={\V} ({\T}(a) \backslash \I{\T}(a)\cdot B_0(a)) \cup {\V}({\T}(a)\cup \{B_0(a)\}\backslash \I{\T}(a))\\
&=(\cup_{{\bf G}\in {\mathbb G}_0}{\V}({\bf G}(a)\backslash B_0(a)) )\cup( {\V}({\T}(a)\backslash \I{\T}(a))\cap {\V}(B_0(a)))\\
&=(\cup_{{\bf G}\in {\mathbb G}_0}{\V}({\bf G}(a)\backslash B_0(a)) )\cup( {\V}({\T}(a)\backslash \I{\T}(a))\cap \cup_{i\in \mathcal{N}}{\V}({\bf C}_i(a)\backslash \I{{\bf C}_i}(a)))\\
&=(\cup_{{\bf G}\in {\mathbb G}_0}{\V}({\bf G}(a)\backslash B_0(a)) ) \cup( \cup_{i\in \mathcal{N}}{\V}({\T}_i(a)\backslash \I{{\T}_i}(a)))\\
&=(\cup_{{\bf G}\in {\mathbb G}_0}{\V}({\bf G}(a)\backslash B_0(a)) ) \cup( \cup_{i\in \mathcal{N}}\cup_{[{\bf G}, H]\in {\mathbb G}_i}{\V}({\bf G}(a)\backslash H(a))).\\
\end{split}
\end{equation*}
In addition, by Lemma \ref{wusinit} and Corollary \ref{wuonepoly},  $B_0(a) \neq 0$. Thus $[{\bf G}, B_0]$ specializes well at $a$ for every ${\bf G}\in {\mathbb G}_0$ according to Algorithm \ref{ALzdtorc}. By the induction hypothesis, we know that $[{\bf G}, H]$ specializes well at $a$ for every $[{\bf G},H]\in {\mathbb G}_i$ for any $i \in \mathcal{N}$. Therefore, the statement (2) in the specification of Algorithm \ref{ALpdtors} holds.
\end{proof}

\begin{remark}
 By Algorithm \ref{ALpdtors}, for any regular system $[{\T}_i,H_i]$ in the  first output of Algorithm \ref{ALpdtors}, $\res{\I{{\T}}_i,{\T}_i}$ is a factor of $H_i$. If $B \in K[U]$ is the second output of Algorithm \ref{ALpdtors}, by Corollary \ref{wuonepoly}, we know that ${\V}_{\overline{K(U)}} (\res{\I{{\T}_i},{\T}_i}) \subset  {\V}_{\overline{K(U)}}(H_i)$  and
${\V}^U(\res{\I{{\T}_i},{\T}_i}) \subset  {\V}^U(H_i) \subset {\V}^U(B)$.
\end{remark}

\subsection{Computing RDU}

We present the main result in this section.  Algorithm \ref{ALsus} shows how to compute a generic regular system decomposition  and the associated RDU of a given  system simultaneously.

\begin{algorithm}\label{ALsus}
\SetAlgoCaptionSeparator{. }
\DontPrintSemicolon
\caption{{\tt RDU}}
\KwIn{ A parametric system  ${\bf P}$  in  $K[U][X]$,  variables $X=\{x_1,\ldots,x_n\}$.}
\KwOut{ [$\mathbb{TH}$, $B$], where\\

        (1) $\mathbb{TH}$ is a  finite set  of regular systems in $K[U][X]$  such that
        ${\V}_{\overline{K(U)}}({\P})=\cup_{[{\T}, H]\in \mathbb{TH}}{\V}_{\overline{K(U)}}({\T}\backslash H)$,

        (2) $B$   a polynomial in $ K[U]$ for any $a\in \overline{K}^d\backslash {\V}^{U}(B)$,
        ${\V}({\P}(a))=\cup_{[{\T}, H]\in \mathbb{TH}}{\V}({\T}(a)\backslash H(a))$ and $[{\T},H]$ specializes well at $a$ for any $[{\T},H]\in  \mathbb{TH}$.}
     Compute a Wu's decomposition $\{{\bf C}_1, \ldots, {\bf C}_m\}$ of $\P$ in $K[U][X]$ by Wu's method\;
     $B$:=1\;
     $\mathbb{TH}$:=$\emptyset$\;
     \For{$i = 1 \to m$}
     {\eIf{${\bf C}_i$ is a contradictory ascending chain}{
     $B$:=$B \cdot op({\bf C}_i)$}{
     W:=${\tt TSToRS}({\bf C}_i,X)$\;
      $\mathbb{TH}$:=$\mathbb{TH}\cup W_1$,
      $B$:=$B\cdot W_2$ }}
      return [$\mathbb{TH}$, $B$]\;
\end{algorithm}

\begin{theorem}
Algorithm \ref{ALsus} terminates correctly.
\end{theorem}
\begin{proof}
 The termination follows from the termination of Algorithm \ref{ALpdtors}. We only need to show the correctness. In fact, the statement (1) in the specification of Algorithm \ref{ALsus} follows from Wu's method and Algorithm \ref{ALpdtors} and the statement (2) follows from Corollary \ref{wudecom} and Algorithm \ref{ALpdtors}.
\end{proof}

\begin{corollary}\label{CO}
Let $\P$ be a parametric polynomial system, ${\tt RDU}({\P},X)= [\mathbb{TH},B]$. Then $\mathbb{TH}$ is stable at any $a\in\overline{K}^d \backslash {\V}^U(B)$.
\end{corollary}
The proof of Corollary \ref{CO} is similar to that of generic zero-dimensional case. For more details, please see \cite{tang}.

\section{Examples and Implementation}\label{sectionexamples}
In this section, we show by examples how our algorithms work. In addition, we run some benchmarks and report comparison to other tools with similar function. 

Example \ref{ex1} is designed to illustrate how Algorithm \ref{ALsus} works.
\begin{example}\label{ex1}
Consider the system
 \[{\P}=\begin{cases}
  (ux+1)z^3+(vy+1)z^2+wxz+1\\
   ux+1
\end{cases}
\]
where $x$, $y$ and  $z$ are variables $(x\prec y \prec z)$ and $u$, $v$ and $w$ are  parameters.
\end{example}

{\em \bf Step 1:} According to the first step of Algorithm \ref{ALsus}, we get {Wu's decomposition} $\mathbb{S}=\{{\bf C}_1,{\bf C}_2,{\bf C}_3\}$ of $\P$ in $\mathbb{R}[u,v,w][x,y,z]$ where ${\bf C}_1=\{ux+1, u+uvyz^2+uz^2-wz\}$, ${\bf C}_2=\{ux+1, vy+1, u-wz\}$ and ${\bf C}_3=\{-v^2u^5w\}$.

{\em \bf Step 2:} Let $\mathbb{TH}= \emptyset$ and $B= 1$.

{\em\bf Step 3:} Because ${\bf C}_1$ and ${\bf C}_2$ are both non-contradictory ascending chains and ${\bf C}_3$ is a contradictory ascending chain,  we need to execute ${\tt TSToRS}({\bf C}_1,[x,y,z]\})$ and ${\tt TSToRS}({\bf C}_2,[x,y,z] )$.

$\qquad$               {\em\bf Step 3.1:} According to ${\tt TSToRS}$, we execute ${\tt ZDToRC}({\bf C}_1,\mvar{{\bf C}_1})$ where $\mvar{{\bf C}_1}=\{x,z\}$. It returns $W=[\{\{ux+1, u+uvyz^2+uz^2u-w z\}\}, u(vy+1)]$. Since $W_2 \notin K[u,v,w]$, we execute ${\tt wusolve}(u(vy+1))$ and it returns ${\bf C}_{11}=\{vy+1\},{\bf C}_{12}=\{u\}$.

$\qquad$               {\em\bf Step 3.2:} Let ${\T}_1={\T} \cup \{y v+1\}$. Now we need to execute ${\tt TSToRS}({\T}_1)$. When we execute ${\tt ZDToRC}({\T}_1,\mvar{{\T}_1})$, it returns $[\emptyset,vu]$. Since $vu \in K[u,v,w]$, the output of ${\tt TSToRS}({\bf C}_1,\{x,y,z\}\})$ is $[\{[\{ux+1,u+uvyz^2+z^2\cdot  u-wz\},u(vy+1)]\}, uvw]$.

$\qquad$               {\em\bf Step 3.3:} Let $\mathbb{TH}:= \mathbb{TH} \cup \{[\{ux+1,u+uvyz^2+z^2\cdot  u-wz\},u(vy+1)]\}$ and $B=B\cdot uvw$.

$\qquad$               {\em\bf Step 3.4:} According to ${\tt TSToRS}$, we execute ${\tt ZDToRC}({\bf C}_2,\mvar{{\bf C}_2})$ where $\mvar{{\bf C}_2}=\{x,y,z\}$. It returns $W=[\{ux+1, vy+1, u-wz\},-uvw]$. Since $W_2 \in K[u,v,w]$,  the output of ${\tt TSToRS}({\bf C}_2,[x,y,z]\})$ is $[\{[\{ux+1, vy+1, u-wz\},-uvw]\}, uvw]$

$\qquad$               {\em\bf Step 3.5:} Let $\mathbb{TH}:= \mathbb{TH} \cup \{[\{ux+1, vy+1, u-wz\},-uvw]\}$ and $B=B\cdot uvw$.

$\qquad$               {\em\bf Step 3.6:} Since ${\bf C}_3$ is a contradictory ascending chain, we execute $B:=B \cdot -v^2u^5w$.

{\em\bf Step 4:} Finally, we get $B=u^2v^2w^2$ and $\mathbb{T}= \{[\{ux+1,u+uvyz^2+uz^2-wz\},u (vy+1)],[\{ux+1, vy+1, u-wz\},-uvw]\}$.

In the example, we factor the polynomials and let the polynomial be squarefree in some steps.

With the result above, we get a regular system decomposition of $\P$  in $\mathbb{R}[u,v,w][x,y,z]$ and a  RDU variety $\mathcal{V}=\{(u,v,w) \in \mathbb{C}^3 | uvw =0 \} $. According to the specification of  Algorithm \ref{ALsus}, for any $a \in \mathbb{C}^3 \backslash \mathcal{V}$, ${\V}(P(a))= \cup_{[{\T},H] \in \mathbb{TH}} {\V}({\T}(a) \backslash H(a))$ where $[{\T}(a), H(a)]$ is a regular system. With the definition of RDU variety, we get that $\mathcal{V}$ is the RDU variety of $\P$ {\em w.r.t} $\mathbb{TH}$.

    Example \ref{ex2} is provided by Changbo Chen, which is a good example to show how the orderings of variables affect the results and the efficiency of Algorithm \ref{ALsus}.





 \begin{example}\label{ex2}
Consider the parametric system
 \[{\P}=\begin{cases}
  d_4d_3r + r_2^2 - d_4d_3r_2^2 + d_4^2d_3^2 - d_4d_3^3 - d_4^3d_3 + d_4d_3 + Z - r)t^4+ (-2r_2d_4r + 2r_2d_4^3 + \\
  2r_2d_4d_3^2- 4r_2d_3d_4^2 + 2r_2^3d_4 + 2r_2d_4)t^3-(2r_2^2- 2r + 4d_4^2r_2^2+4d_4^2 + 2Z - 2d_4^2d_3^2)t^2 +\\
   (-2r_2d_4r + 2r_2d_4d_3^2 + 2r_2d_4 + 2r_2d_4^3 + 4r_2d_3d_4^2 + 2r_2^3d_4)t + r_2^2 + d_4^3d_3 - d_4d_3r + \\
   d_4d_3r_2^2 + Z - r - d_4d_3 + d_4^2d_3^2 + d_4d_3^3

\end{cases}
\]
where $r$, $Z$, $t$ are variables and $r_2$, $d_3$, $d_4$ are  parameters.
\end{example}

The $3$ variables can be ordered in $6$ different ways. We tried all these orders when calling Algorithm \ref{ALsus} for this example. We are interested in the RDU varieties output by the algorithm, so we only report the second output of the algorithm.


By calling ${\tt RDU}({\P},[r,t,Z])$ and ${\tt RDU}({\P},[t,r,Z])$ , we get

  \[{\tt RDU}({\P},[r,t,Z])_2= r_2d_4 \]
  \[{\tt RDU}({\P},[t,r,Z])_2= r_2d_4 \]

By calling ${\tt RDU}({\P},[Z,r,t])$, we get

\[ {\tt RDU}({\P},[Z,r,t])_2=(-d_3+d_4r_2^2+d_4)d_4r_2(d_4d_3-1) \]

By calling ${\tt RDU}({\P},[r,Z,t])$, we get
\[{\tt RDU}({\P},[r,Z,t])_2= r_2d_4(d_3-d_4r_2^2-d_4) \]

 By calling ${\tt RDU}({\P},[t,Z,r])$, we get

\[{\tt RDU}({\P},[t,Z,r])_2= r_2d_4(d_3^2-3r_2^2)(d_4d_3-1) \]

${\tt RDU}({\P},[Z,t,r])_2$  is too huge to be listed here. It has 11 factors and contains 10838 terms (after expanding).

\begin{Table}\label{exam1}
\begin{center}
{Timings (in second) of Example \ref{ex2} under different orders of variables}\par
\begin{tabular}{|c|c|c|c|c|c|c|}
  \hline
   & [r,t,Z] &[t,r,Z] &[Z,r,t] &[r,Z,t] &  [t,Z,r]&   [Z,t,r]  \\
   \hline
  time &0.016 &0.015 &0.016 & 0.046  &  0.016&  3.931  \\
  \hline
\end{tabular}
\end{center}
\end{Table}

The timings for the above computation are shown in Table \ref{exam1}. How to choose a suitable order in advance is an interesting topic for our future work.











\begin{Table}\label{res1}
\begin{center}
     {Comparing {\tt RDU} and \tt{Triangularize}}\\
     \begin{tabular}{|c|c|c|c|c|c|c|c|}
     \hline
      number&system &$U$ &$X$ &$\tt{wusolve}$ & $\tt{TSToRS}$&total&$\tt{Triangularize}$ \\
     \hline
     \hline
      1. &{\em Hereman-2} &1&7 & 0.063& 0. &0.063&0.452\\
      2. &{\em Hereman-8-8}  &3&5&0.281&0.015& 0.296&0.188\\
      3. &{\em Maclane} &3&7& 0.093	&0.078&	0.171&	0.156\\
      4.&{\em MontesS7}&1&3&	0.967&	0.015&	0.982&	0.063\\

      5. &{\em MontesS11}&3&3&0.016	&0.&	0.016&	0.\\
      6. &{\em MontesS12}&2&6&0.109&	0.078&	0.187&	0.093\\
      7. &{\em MontesS13}&3&2&0.016&	0&	0.016	&0.031\\

      8. &{\em MontesS14} &1&4&0.031	&0.016	&0.047&	0.141\\
       9.  &{\em MontesS15}&4&8&0.&0.&0.	&0.016\\
        10. &{\em MontesS16}&3&12&0.016	&0.&	0.016&	0.156\\
        11.&{\em MontesS18}&2&3&0.296	&0.031&	0.327&	0.811\\
        12.&{\em AlkashiSinus}&3&6&0.&	0.	&0.&	0.047\\
        13.&{\em Bronstein}&2&2 &0.016&	0.&	0.016&	0.062\\
        14.&{\em Cheaters-homotopy-easy}&4&3&0.358	&10.749&	11.107&	0.047\\
        15.&{\em Cheaters-homotopy-hard}&5&2&0.&	32.867&	32.867	&0.062\\
        16.&{\em Gerdt}&3&4&0.015&	0.&	0.015&	0.016\\
        17.&{\em Lanconelli}&7&4&0.047&	0.&	0.047&	0.\\
       18 .&{\em Lazard-ascm2001}&3&4&0.811	&0.047&	0.858	&0.39\\
       19.&{\em Leykin-1}&4&4&	0.312&	0.&	0.312&	0.531\\
       20.&{\em Neural	}&1&3&0.047&	0.016&	0.063	&0.093\\
        21.&{\em Pavelle}&4&4&0.109&	0.062	&0.171&	0.203\\
    22. &{\em SY14}  &2&2&0.016&0. &0.016&0.\\
     23. &{\em Wang93}  &2&3&0.031&0. &0.031&0.046\\
    24.  &{\em zhou3}  &6&11&0.062	&0.047&	0.109&	0.281\\
    25. &{\em zhou4}  &4&7&0.016&	0.016&	0.032&	0.093\\
    26.&{\em KdV}&15&11&0.406&	0.&	0.406&	0.031\\
    27.&P3P&5&2&	0	&0.031&	0.031&	0.047\\
    28.&{\em SBCD23}&1&3&	0.031&	0.&	0.031&	0.047\\
    29.&{\em SBCD24}&1&4&	0.655&	0.016&	0.671&	0.312\\
     \hline
    \end{tabular}
\end{center}
   \end{Table}

We have implemented our algorithms with Maple 16. More specifically, Wu's method for computing parametric triangular decompositions introduced in Section \ref{sfuhaoshuoming} is implemented as a function ${\tt wusolve}$ and Algorithm \ref{ALPRSD} is implemented as a function ${\tt WRSD}$. We ran many examples collected from other papers \cite{Montes,sun,changbo}. At the same time, we compare the running time with {\tt Triangularize}\footnote{Please find more details on Triangularize from the help document of Maple 16.} in {\tt RegularChains} which can compute regular decompositions of given polynomial systems with parameters.
Throughout this section,  all the results are obtained in Maple 16 using an Intel(R) Core(TM) i5 processor (3.20GHz CPU and 2.5 GB total memory) and Windows 7 (32 bit). The empirical data about timings is presented in Table \ref{res1}.

In Table \ref{res1},  the column marked $U$ and $X$ mean the number of parameters and variables, respectively. The column marked $\tt{wusolve}$ means the time used by $\tt{wusolve}$ at the first step in $\tt{RDU}$ (see Algorithm \ref{ALsus}). The column marked $\tt{TSToRS}$ means the time used by $\tt{TSToRS}$. The column marked $\tt{Triangularize}$ means the time used by $\tt{Triangularize}$.  Some data which shows $0.$ means that the data is less than $0.001$ and ignored by system.

From  Table \ref{res1}, we find that the cost of $\tt{wusolve}$ takes up a majority of the total time in most examples when solving practical problems as shown. Comparing with {\tt Triangularize}, for some examples, our algorithm is faster than {\tt Triangularize}.   As for the example {\em Cheaters-homotopy-easy} and {\em Cheaters-homotopy-hard}, the time used by {\tt WRSD} seems too much. This is not reasonable and optimization should be done in the future. 

\section{Conclusions}\label{con1}

The focus of this paper is how to decompose a parametric system into regular systems at ``parameter level" and the property of this decomposition under specification. We provide an algorithm for computing GRDs of parametric systems and the related RDU varieties simultaneously no matter the systems are generic zero-dimensional or positive-dimensional, which is a generalization of our earlier work in \cite{tang} for generic zero-dimensional case. Then any parametric system in $K[U][X]$ can be decomposed into finitely many regular systems and the decomposition is stable at any parameter value in the complement of the associated RDU variety of the parameter space. That is to say, once we obtain such decomposition, all the solutions of the original system are expressed by some regular systems, except for some possible solutions over the parameter values on the associate RDU variety.

Note that, using  Algorithm \ref{ALsus} of generic regular decomposition, we may get a complete decomposition of a given parametric system ${\P}$ step by step. A rough procedure is as follows. Firstly, for certain variables $X$ and parameters $U$, we call ${\tt RDU}({\P},X)$ and get a set of regular systems and a parametric polynomial $B \in K[U]$. Then, let $U_1= U\backslash \{u\}$  and $X_1= X \cup \{u\}$ where $u \in U$ and we call ${\tt RDU}({\P} \cup \{B\},X_1)$ to get a set of regular systems and a parametric polynomial $B_1 \in K[U_1]$. Continuing this iteration, until $U_i=\emptyset$ and $B_i \in K$ for some $i$,  we can get a complete decomposition finally. This is what we called {\em hierarchical strategy} in Section \ref{SecIntro}. What can we  benefit from this strategy? According to this method, we can stop at any step of the iteration especially when the computation is hard to be finished. Then we get a generic regular decomposition and a set of parametric polynomials which can determinate a RDU variety (low dimensional variety). For some huge problems, limited by the computational capacity of micro computer, we can get a partial solution, which is useful if one cannot get any information by other complete methods. Of course, the procedure should be described clearly and proved to be correct. That is one of our future work.

\section*{Acknowledgements}
The work is partly supported by the National Natural Science Foundation of China (Grant No.11271034, No.11290141)
and the project SYSKF1207 from SKLCS, IOS, the Chinese Academy of Sciences.   Thank Changbo  Chen and Yao Sun for providing a great deal of test-systems.

\end{document}